\documentclass{article}
\usepackage{amssymb,amsmath}
\usepackage{graphicx,enumerate}

\title{Optimal prevention with possibilistic and mixed background risk}

\author{Irina Georgescu \\ \footnotesize Academy of Economic Studies\\ \footnotesize Department of Economic Informatics and Cybernetics\\ \footnotesize Pia$\c{t}$a Romana No 6  R 70167, Oficiul Postal 22, Bucharest, Romania\\
 \footnotesize Email: irina.georgescu@csie.ase.ro \\ and \\ Ana Mar\'{i}a Lucia Casademunt\\ \footnotesize Universidad Loyola--Andaluc\'{i}a\\ \footnotesize
Department of Business Administration,
Cordoba, Spain\\
\footnotesize Email: alucia@uloyola.es}

\date{}

\begin{document}
\maketitle

\begin{abstract}
In this paper the effect of posibilistic or mixed background risk on the level of optimal prevention is studied. In the framework of five purely possibilistic or mixed models, necessary and sufficient conditions are found such that the level of optimal saving decreases or increases as a result of the actions
of various types of background risk. This way our results complete those obtained by Courbage and Rey for some prevention models with probabilistic background risk.

\end{abstract}


\newtheorem{definitie}{Definition}[section]
\newtheorem{propozitie}[definitie]{Proposition}
\newtheorem{remarca}[definitie]{Remark}
\newtheorem{exemplu}[definitie]{Example}
\newtheorem{intrebare}[definitie]{Open question}
\newtheorem{lema}[definitie]{Lemma}
\newtheorem{teorema}[definitie]{Theorem}
\newtheorem{corolar}[definitie]{Corollary}

\newenvironment{proof}{\noindent\textbf{Proof.}}{\hfill\rule{2mm}{2mm}\vspace*{5mm}}

\section{Introduction}

The prevention is the effort made by an agent to reduce the action of an unwanted event (e. g. a financial loss). The effort can be the amount of money invested by the agent in prevention activities and then he is interested to determine the optimal level of this effort.

According to \cite{courbage2}, the prevention consists in reducing the probability of a loss (self--protection), the size of loss (self--insurance) or both of them (self--insurance--cum--protection). For simplicity, in this paper by prevention one understands self--protection (this terminology is used in \cite{courbage1}).

The first contribution to the study of optimal prevention is the paper \cite{ehrlich} by Elrich and Becker which investigates the interconnection between the market insurance, self-insurance and self--protection. Then, the subject generated a real literature focused especially on the way consumers' preferences influence the optimal level of prevention (see \cite{courbage2}).

A common topic deals with the relationship between risk aversion and optimal prevention. For instance, Dionne and Eeckhoudt emphasize in \cite{dionne}   the fact that the increase of risk aversion does not lead to an increase in optimal investment in prevention, its effect being ambiguous.
A second topic is the relationship between the notion of prudence introduced by Kimball in \cite{kimball} and that of optimal prevention (see \cite{eeckhoudt2}, \cite{courbage2}).
Paper \cite{eeckhoudt2} studies the impact of prudence on optimal prevention in a one--period model, while \cite{courbage2} analyzes the obtaining of the optimal prevention in a
two-period framework. In the latter case, the prevention effort is done in the first period, and its effect on the likelihood of appearance of loss occurs during the second period.

Prevention models with background risk appear for the first time in the paper by Courbage and Ray \cite{courbage1} and Eeckhoudt et al. \cite{eeckhoudt3}. Background risk may appear in one or both periods. With respect to this, the changes in optimal prevention level occur in relation to a benchmark model (in which the background risk does not exist).

In the mentioned papers we deal with probabilistic models, in which the background risk is represented by a random variable. This paper aims to study optimal prevention in non--probabilistic models, in which we have a possibilistic or mixed background risk (see \cite{casademunt}, \cite{georgescu1}, \cite{georgescu2}). In the former case, the two--period background risk will be a fuzzy number, and in the latter case, in one period, background risk will be a random variable, and in the other period, it will be a fuzzy number.

Section $2$ contains the definition of possibilistic expected utility and a few of its properties which will be used in this paper.

In Section $3$ the benchmark prevention model is presented following \cite{courbage1}, \cite{courbage2}. This will be the starting point in the background risk model construction from the next sections.

In Section $4$ eight background risk prevention models are described. These are two--period constructions, and the risk from the two periods may be either a random variable or a fuzzy number. The first three models are probabilistic, and their theory has been already developed in \cite{courbage1}. The other five models are purely possibilistic or mixed. Their construction follows a similar line to the first three: the total utility functions are defined, then first order conditions verified by optimal prevention are deduced.

Section $5$ focuses on a model in which various types of background risk of the five models influence the level of optimal prevention. For each of the five models, the level of optimal prevention is compared to that of the benchmark model. Necessary and sufficient conditions for the increase or decrease in optimal prevention with respect to the existing level in case of the benchmark model are proved.

In Section $6$ the background risk effect on optimal saving is studied when to the already existing risk in one period a (possibilistic or probabilistic) risk is added in the other period.

\section{Possibilistic expected utility}

In this section we recall the definition of possibilistic expected utility and a few of its basic properties (cf. \cite{carlsson1}, \cite{carlsson2}, \cite{dubois1}, \cite{dubois2}, \cite{fuller}).

If $X$ is a random variable defined on a probability field $(\Omega, \mathcal{K}, P)$ with $\Omega \subseteq {\bf{R}}$ then we will denote by $M(X)$ its expected value. For any continuous utility function $u: {\bf{R}} \rightarrow {\bf{R}}$, the expected value $M(u(X))$ of the random variable $u(X)$ is the probabilistic expected utility associated with $X$ and $u$.

To develop a possibilistic risk theory we need a similar notion of possibilistic expected utility. Its definition assumes a framework composed of the following components:

$\bullet$ a utility function $u: {\bf{R}} \rightarrow {\bf{R}}$ representing an agent;

$\bullet$ a fuzzy number $A$ modelling the possibilistic risk;

$\bullet$ a weighting function, i.e. a non--negative and monotone increasing function $f: [0, 1] \rightarrow {\bf{R}}$ satisfying the normality condition
$\int_0^1 f(\gamma)d\gamma=1$.

Assume that the level sets of the fuzzy number $A$ have the form $[A]^\gamma=[a_1(\gamma), a_2(\gamma)]$, for any $\gamma \in [0, 1]$. Then the possibilistic expected utility $E_f(u(A))$ is defined by

(2.1) $E_f(u(A))=\frac{1}{2} \int_0^1 [u(a_1(\gamma))+u(a_2(\gamma))]f(\gamma)d\gamma$

If $u$ is the identity function of $\bf{R}$ then $E_f(u(A))$ is exactly the possibilistic expected value \cite{carlsson1}, \cite{fuller}.

(2.2) $E_f(A)=\frac{1}{2} \int_0^1 [a_1(\gamma)+a_2(\gamma)]f(\gamma)d\gamma$.

The following two propositions emphasize two important properties of possibilistic expected utility: the first one is a linearity property, the second one is a Jensen type inequality (used in the following in the analysis of optimal prevention variations with respect to various types of background risk).

\begin{propozitie}\cite{georgescu1}
Let $g: {\bf{R}} \rightarrow {\bf{R}}$, $h: {\bf{R}} \rightarrow {\bf{R}}$ two continuous utility functions and $a, b \in {\bf{R}}$. If $u=ag+bh$, then $E_f(u(A))=aE_f(g(A))+bE_f(h(A))$.
\end{propozitie}

\begin{propozitie}\cite{georgescu1}, \cite{casademunt}
If the utility function $u$ is convex, then $u(E_f(A)) \leq E_f(u(A))$.
\end{propozitie}

\section{The benchmark prevention model}

In this section we present, according to \cite{courbage1}, \cite{courbage2}, the entities which compose the benchmark prevention model, the form of its total utility function and the first--order condition with which the level of optimal prevention is selected.

The benchmark prevention model \cite{courbage1}, \cite{courbage2} has two periods $1$ and $2$: the prevention effort is conducted in period $1$, and the effect of this effort on the probability of the occurrence of a loss in wealth) is manifested in period $2$.

The model is defined by the following entities:

$\bullet$ $w_1, w_2$ are the safe wealth in periods $1$ and $2$;

$\bullet$ $u$ and $v$ are utility functions for period $1$ and $2$, respectively;

$\bullet$ $e$ is the level of prevention effort (in period $1$);

$\bullet$ $l$ is the loss level;

$\bullet$ $p(e)$ is the probability that a loss $l$ occurs at the prevention level $e$.

In \cite{courbage1}, \cite{courbage2} this model is studied under the hypotheses:

$\bullet$ the utility functions $u$ and $v$ are of class $\mathcal{C}^2$, $u'>0$, $v'>0$, $u''<0$, $v''<0$;

$\bullet$ $p'(e)<0$, $p''(e)>0$, for any level $e$;

$\bullet$ $l>0$.

According to these initial data, in the second period of the consumption activity we will have an endogenous risk $\tilde \epsilon_e$ (parameterized by the prevention level $e$), whose outcomes are $-l$ and $0$:

\[
\tilde \epsilon_e=
  \begin{pmatrix}
    -l & 0  \\
    p(e)& 1-p(e)
  \end{pmatrix}
\]

Then the expected utility corresponding to the second period will be $M(w_2+\tilde \epsilon_e)=p(e)v(w_2-l)+(1-p(e))v(w_2)$, and the total utility function will be:

(3.1) $V(e)=u(w_1-e)+p(e)v(w_2-l)+(1-p(e))v(w_2)$.

The agents wish to choose a level $e$ of the effort maximizing the total utility $V(e)$. This comes to finding the solution
$e^\ast$ of the optimization problem:

$\displaystyle \max_{e} V(e)$.

Since $V(e)$ is strictly concave, the optimal solution $e^\ast$ is given by $V'(e^\ast)=0$, which leads us to the first--order condition:

(3.2) $u'(w_1-e^\ast)=p'(e^\ast)[v(w_2-l)-v(w_2)]$.

\section{Prevention models with background risk}

The prevention models from this section are built starting from the benchmark model from Section $3$. Their definition is founded on the following assumptions:

{\emph{Assumption 1}} Background risk may appear both in period $1$ and in period $2$.

{\emph{Assumption 2}} The background risk from each period can be modeled both probabilistically (by a random variable), and possibilistically (by a fuzzy number).

Denoting by $0$ the lack of background risk in one of the two periods, the eight models which can be defined on the basis of the two assumptions are specified in Table $1$:

\begin{table}[ht]
\centering
\begin{tabular}{c c c c }
\hline \hline
   & Period 1 & Period 2  \\[0.5ex]
\hline
1 & probabilistic background risk  & 0 \\
2 & 0  & probabilistic background risk  \\
3 & probabilistic background risk  & probabilistic background risk \\
4 & possibilistic background risk  & 0  \\
5 & 0 & possibilistic background risk  \\
6 & possibilistic background risk  & possibilistic background risk   \\
7 & probabilistic background risk & possibilistic background risk   \\
8 & possibilistic background risk  & probabilistic background risk  \\[1ex]
\hline
\end{tabular}
\label{table: nonlin}
\end{table}

In cases $1-3$ we deal with probabilistic models, in cases $4-6$ with possibilistic models, and in the last two cases with mixed models.

Models $1-3$ have been studied in \cite{courbage1}. Next we present briefly the construction of these models.

{\emph{Model 1}} \cite{courbage1}
In period $1$ appears a probabilistic risk modeled by the random variable $X$. The total utility function  has the following form:

$V_1(e)=M(u(w_1-e+X))+p(e)v(w_2-l)+(1-p(e))v(w_2)$.

The first derivative of $V_1(e)$ will be:

$V'_1(e)=-M(u'(w_1-e+X))+p'(e)[v(w_2-l)-v(w_2)]$

and the optimal value $e_1^\ast$ will be given by the first--order condition:

(4.1) $M(u'(w_1-e_1^\ast+X))=p'(e_1^\ast)[v(w_2-l)-v(w_2)]$.

{\emph{Model 2}} \cite{courbage1}
In period $2$ appears a probabilistic risk modeled by the random variable $Y$. The overall utility function will be

$V_2(e)=u(w_1-e)+p(e)M(v(w_2-l+Y))+(1-p(e))M(v(w_2+Y))$

The first derivative of $V_2(e)$ will be:

$V'_2(e)=-u'(w_1-e)+p'(e)[M(v(w_2-l+Y))-M(v(w_2+Y))]$.

and the optimal value $e_2^\ast$ of $e$ will be given by the condition:

(4.2) $u'(w_1-e_2^\ast)=p'(e_2^\ast)[M(v(w_2-l+Y))-M(v(w_2+Y))]$.

{\emph{Model 3}} \cite{courbage1}
Background risk appears in both periods and is modeled by the random variables $X$, and $Y$, respectively. The total utility function has the form:

$V_3(e)=M(u(w_1-e+X))+p(e)[M(v(w_2-l+Y))-M(v(w_2+Y))]$.

Then the first derivative of $V_3(e)$ will be:

$V'_3(e)=-M(u'(w_1-e+X))+p'(e)[M(v(w_2-l+Y))-M(v(w_2+Y))]$

and the optimal value $e_3^\ast$ of $e$ will be given by the condition:

(4.3) $M(u'(w_1-e_3^\ast+X))=p'(e_3^\ast)[M(v(w_2-l+Y))-M(v(w_2+Y))]$

Next we will present the construction of models $4-8$ from Table 1. For the rest of the section we fix a weighting function $f: [0, 1] \rightarrow {\bf{R}}$.

{\emph{Model 4}} Assume that the background risk is described by a fuzzy number $A$ whose level sets are $[A]^\gamma=[a_1(\gamma), a_2(\gamma)]$ for any $\gamma \in [0, 1]$. The total utility function of the model will be:

$V_4(e)=E_f(u(w_1-e+A))+p(e)v(w_2-l)+(1-p(e))v(w_2)$

According to the definition (2.1) of possibilistic expected utility, $V_4(e)$ can be written:

$V_4(e)=\frac{1}{2} \int_0^1 [u(w_1-e+a_1(\gamma))+u(w_1-e+a_2(\gamma))]f(\gamma)d\gamma+p(e)v(w_2-l)+(1-p(e))v(w_2)$.

By derivation, one obtains

$V'_4(e)=-\frac{1}{2} \int_0^1 [u'(w_1-e+a_1(\gamma))+u'(w_1-e+a_2(\gamma))]f(\gamma)d\gamma+p'(e)[v(w_2-l)-v(w_2)]$,

from where, by (2.1), the following form of the derivative of $V_4(e)$ follows:

(4.4) $V'_4(e)=-E_f(u'(w_1-e+A))+p'(e)[v(w_2-l)-v(w_2)]$.

By deriving one more time, one reaches

$V''_4(e)=E_f(u''(w_1-e+A))+p''(e)[v(w_2-l)-v(w_2)]$.

From the hypotheses imposed on the utility functions $u$ and $v$ on the probability $p(e)$ it follows $V''_4(e)<0$ for any level $e$ of the prevention effort, thus $V_4$ is a concave function. Then the value  $e_4^\ast$ of $e$ maximizing the total utility function $V_4$ verifies the first--order condition $V'_4(e_4^\ast)=0$, which, by (4.4), is written as:

(4.5) $E_f(u'(w_1-e_4^\ast+A))=p'(e_4^\ast)[v(w_2-l)-v(w_2)]$.

{\emph{Model 5}} Assume that the background risk in period $2$ is a fuzzy number $B$ whose level sets are $[B]^\gamma=[b_1(\gamma), b_2(\gamma)]$. The total utility function of the model will be:

$V_5(e)=u(w_1-e)+p(e)E_f(v(w_2-l+B))+(1-p(e))E_f(v(w_2+B))$

$=u(w_1-e)+\frac{1}{2} p(e)\int_0^1 [v(w_2-l+b_1(\gamma))+v(w_2-l+b_2(\gamma))]f(\gamma)d\gamma+\frac{1}{2} (1-p(e))\int_0^1[v(w_2+b_1(\gamma))+v_2(w_2+b_2(\gamma))]f(\gamma)d\gamma$.

By derivation one obtains:

(4.6) $V'_5(e)=-u'(w_1-e)+p'(e)[E_f(v(w_2-l+B))-E_f(v(w_2+B))]$

$V''_5(e)=u''(w_1-e)+p''(e)[E_f(v(w_2-l+B))-E_f(v(w_2+B))]$

Due to the monotonicity property of the possibilistic expected utility the inequality $E_f(v(w_2-l+B))-E_f(v(w_2+B))\leq 0$, therefore $V''_5(e)<0$ for any $e$, hence $V_5$ is strictly concave. Then the value $e_5^\ast$ of $e$ maximizing $V_5(e)$ verifies the condition $V'_5(e_5^\ast)=0$, which, by (4.6), can be written:

(4.7) $u'(w_1-e^\ast_5)=p'_5(e_5^\ast)[E_f(v(w_2-l+B))-E_f(v(w_2+B))]$.

{\emph{Model 6}} Assume that the background risk in period $1$ is the fuzzy number $A$, and the background risk from period $2$ is the fuzzy number $B$.
The total utility function of the model will be:

$V_6(e)=E_f(u(w_1-e+A))+p(e)E_f(v(w_2-e+B))+(1-p(e))E_f(v(w_2+B))$.

By derivation one obtains:

(4.8) $V'_6(e)=-E_f(u'(w_1-e+A))+p'(e)[E_f(v(w_2-l+B))-E_f(v(w_2+B))]$

$V''_6(e)=E_f(u''(w_1-e+A))+p''(e)[E_f(v(w_2-l+B))-E_f(v(w_2+B))]$

From the last inequality $V''_6(e)\leq 0$ follows for any $e$, thus $V_6$ is strictly concave. Then the value $e_6^\ast$ of $e$ maximizing $V_6(e)$ verifies the condition $V'_6(e_6^\ast)=0$, which, by (4.8), can be written:

(4.9) $E_f(u'(w_1-e_6^\ast+A))=p'(e)[E_f(v(w_2-l+B))-E_f(v(w_2+B))]$

{\emph{Model 7}} Assume that the background risk in period $1$ is the random variable $X$, and the background risk in period $2$ is the fuzzy number $B$. The total utility function of the model will be:

$V_7(e)=M(u(w_1-e+X))+p(e)E_f(v(w_2-l+B))+(1-p(e))E_f(v(w_2+B))$

By derivation one obtains

(4.10) $V'_7(e)=-M(u'(w_1-e+X))+p'(e)[E_f(v(w_2-l+B))-E_f(v(w_2+B))]$

It follows $V'(e)<0$ for any level $e$ of the effort level, thus the value $e_7^\ast$ of $e$ maximizing $V_7(e)$ is given by $V'_7(e_7^\ast)=0$. By (4.10), this first--order condition is written as:

(4.11) $M(u'(w_1-e_7^\ast+X))=p'(e_7^\ast)[E_f(v(w_2-l+B))-E_f(v(w_2+B))]$.

{\emph{Model 8}} Assume that the background risk in period $1$ is the fuzzy number $A$, and the background risk in period $2$ is a random variable $Y$. The total utility function of the model will be:

$V_8(e)=E_f(u(w_1-e+A))+p(e)M(v(w_2-l+Y))+(1-p(e))M(v(w_2+Y))$

By derivation one obtains

(4.12) $V'_8(e)=-E_f(u'(w_1-e+A))+p'(e)[M(v(w_2-l+Y))-M(v(w_2+Y))]$

One proves immediately that $V_8$ is concave, thus the value $e_8^\ast$ of $e$ maximizing $V_8(e)$ verifies $V'_8(e_8^\ast)=0$, which, by (4.12), can be written:

(4.13) $E_f(u'(w_1-e_8^\ast+A))=p'(e_8^\ast)[M(v(w_2-l+Y))-M(v(w_2+Y))]$

\begin{remarca}
The above models incorporate both the entities of the benchmark model, and the features of background risk types defining them. These appear in the expression of total utility functions and in the first--order condition which determines the level of optimal prevention.
\end{remarca}

\section{The effect of background risk on optimal prevention}

In this section we study the way adding various types of background risk may produce changes in the level of optimal saving. For this we have to compare
the optimal prevention level $s^\ast$ of the benchmark model with the optimal prevention levels $s_i^\ast, i=1,\ldots, 8$ of the eight models of the previous section. The models $1-3$ being already treated in \cite{courbage1}, we will handle the models $4-8$ in Table $1$. For each case, we will be situated in the framework set defined in Section $4$.

{\emph{Case 4}} (the comparison of $e^\ast$ and $e_4^\ast$)

The following result establishes a necessary and sufficient condition such that, adding a possibilistic background risk $A$ in period $1$, the optimal saving level drops.

\begin{propozitie}
The following assertions are equivalent:

(i) $e_4^\ast \leq e^\ast$;

(ii) $u'(w_1-e^\ast)\leq E_f(u'(w_1-e^\ast+A))$.

\end{propozitie}

\begin{proof}
From the previous section we know that $V'_4(e)<0$ for any $e$, thus $V'_4$ is a strictly decreasing function. Then, by (4.4) and (3.2), the following equivalences hold:

$e_4^\ast \leq e^\ast$ iff $V'_4(e^\ast) \leq V'_4(e_4^\ast)=0$;

\hspace{1.2cm} iff $-E_f(u'(w_1-e^\ast+A))+p'(e^\ast)[v(w_2-l)-v(w_2)]\leq 0$;

\hspace{1.2cm} iff $-E_f(u'(w_1-e^\ast+A))+u'(w_1-e^\ast)\leq 0$;

\hspace{1.2cm} iff $u'(w_1-e^\ast)\leq E_f(u'(w_1-e^\ast+A))$.
\end{proof}

\begin{corolar}
If $u'''>0$ and $E_f(A)=0$ then $e_4^\ast \leq e^\ast$.
\end{corolar}

\begin{proof}
By applying Proposition 2.2 to the convex function $u'$ it follows that $u'(w_1-e)=u'(E_f(w_1-e+A))\leq E_f(u'(w_1-e+A))$, hence by Proposition 5.1, one obtains $e_4^\ast \leq e$.
\end{proof}

Corollary 5.2 shows that if the agent is prudent \footnote{The notion of prudence in a possibilistic or mixed framework was studied in \cite{casademunt}, \cite{dionne}}(in period $1$) then the presence of possibilistic background risk in the first period reduces the optimal prevention.

{\emph{Case 5}} (the comparison of $e^\ast$ and $e_5^\ast$)

The following result establishes a necessary and sufficient condition such that, adding a possibilistic background risk $B$ in period $2$, the optimal prevention level raises.

\begin{propozitie}
The following assertions are equivalent:

(i) $e_5^\ast \geq e^\ast$;

(ii) $v(w_2-l)-v(w_2)-E_f(v(w_2-l+B))+E_f(v(w_2+B))\geq 0$.
\end{propozitie}

\begin{proof}
The function $V'_5$ is strictly decreasing, thus by applying (4.6) and (3.2), the following equivalences hold:

$e_5^\ast \geq e^\ast$ iff $V'_5(e^\ast) \leq V'_5(e_5^\ast)=0$;

\hspace{1.2cm} iff $-u'(w_1-e^\ast)+p'(e^\ast)[E_f(v(w_2-l+B))-E_f(v(w_2+B))]\leq 0$;

\hspace{1.2cm} iff $-p'(e^\ast)[v(w_2-l)-v(w_2)-E_f(v(w_2-l+B))+E_f(v(w_2+B))] \geq 0$;

\hspace{1.2cm} iff $v(w_2-l)-v(w_2)-E_f(v(w_2-l+B))+E_f(v(w_2+B)) \geq 0$.

\end{proof}

\begin{lema}
If $v'''>0$ and $E_f(B)=0$ then $v(w_2-l)-w(w_2) \geq E_f(v(w_2-l+B))-E_f(v(w_2+B))$.
\end{lema}

\begin{proof}
We consider the function $h(x)=v(x)-v(x-l)$, for any $x \in {\bf{R}}$. From the assumption $v'''>0$ it follows that $v''$ is strictly increasing, therefore $h''(x)=v''(x)-v''(x-l)>0$, for any $x$. It follows that the function $h$ is convex, thus, by applying Proposition 2.2, $h(w_2)=h(E_f(w_2+B))\leq E_f(h(w_2+B))$. Then, by noticing that $h(w_2)=v(w_2)-v(w_2-l)$ and $E_f(h(w_2+B))=E_f(v(w_2+B))-E_f(v(w_2-l+B))$ one obtains the desired inequality.
\end{proof}

\begin{corolar}
If $v'''>0$ and $E_f(B)=0$ then $e_5^\ast \geq e^\ast$.
\end{corolar}

\begin{proof}
Proposition 5.3 and Lemma 5.4 are applied.
\end{proof}

{\emph{Case 6}} (the comparison of $e^\ast$ and $e_6^\ast$)

The following result establishes a necessary and sufficient condition such that, adding a possibilistic background risk $A$ in period $1$, and a possibilistic background risk $B$ in period $2$, the optimal prevention level raises.

\begin{propozitie}
The following assertions are equivalent:

(i) $e^\ast \leq e^\ast_6$;

(ii) $\frac{E_f(u'(w_1-e^\ast+A))}{u'(w_1-e^\ast)}\leq \frac{E_f(v(w_2-l+B))-E_f(v(w_2+B))}{v(w_2-l)-v(w_2)}$.
\end{propozitie}

\begin{proof}
The function $V'_6$ is strictly decreasing, thus, by applying (4.8) and (3.2), the following assertions are equivalent:

$\bullet$ $e^\ast \leq e^\ast_6$;

$\bullet$ $V'_6(e^\ast) \geq V'_6(e^\ast_6)=0$;

$\bullet$ $-E_f(u'(w_1-e^\ast+A))+p'(e^\ast)[E_f(v(w_2-l+B))-E_f(v(w_2+B))]\geq 0$;

$\bullet$ $-E_f(u'(w_1-e^\ast+A))+u'(w_1-e^\ast)\frac{E_f(v(w_2-l+B))-E_f(v(w_2+B))}{v(w_2-l)-v(w_2)} \geq 0$.

By taking into account that $v(w_2-l)-v(w_2) <0$ it follows immediately the equivalence of the assertions (i) and (ii).
\end{proof}

\begin{remarca}
Assume that $u'''>0$, $v'''>0$, $E_f(A)=0$ and $E_f(B)=0$. Then by a reasoning analogous to the one of Corollary 5.2, from $u'''>0$
and $E_f(A)=0$ it follows:

$\frac{E_f(u'(w_1-e^\ast+A))}{u'(w_1-e^\ast)} \geq 1$.

By applying Lemma 5.4, from $v'''>0$ and $E_f(B)=0$ it follows

$\frac{E_f(v(w_2-l+B))-E_f(v(w_2+B))}{v(w_2-l)-v(w_2)} \geq 1$.

From these two inequalities and Proposition 5.6, we conclude that for a prudent agent in both periods, the effect of possibilistic background risks $A$ and $B$ on optimal prevention is ambiguous: we cannot establish a positioning of the levels $e^\ast$ and $e^\ast_6$ one to each other.
\end{remarca}

{\emph{Case 7}} (the comparison of $e^\ast$ and $e_7^\ast$)

The following result establishes a necessary and sufficient condition, that adding a probabilistic risk $X$ in period $1$ and a possibilistic risk $B$ in period $2$, the optimal prevention level raises.

\begin{propozitie}
The following assertions are equivalent:

(i) $e^\ast \leq e^\ast_7$;

(ii) $\frac{M(u'(w_1-e^\ast+X))}{u'(w_1-e^\ast)}\leq \frac{E_f(v(w_2-l+B))-E_f(v(w_2+B))}{v(w_2-l)-v(w_2)}$.
\end{propozitie}

\begin{proof}
Analogous to the proof of Proposition 5.6 by using the fact that $V'_7$ is strictly decreasing and applying (4.10) and (3.2).
\end{proof}

\begin{remarca}
Assume that $u'''>0$, $v'''>0$, $M(X)=0$ and $E_f(B)=0$. From $u'''>0$ and $M(X)=0$  it follows according to Jensen's inequality, the following inequality:

$\frac{M(u'(w_1-e^\ast+X))}{u'(w_1-e^\ast)} \geq 1$.

By Lemma 5.4, conditions $v'''>0$, $E_f(B)=0$ imply the inequality:

$\frac{E_f(v(w_2-l+B))-E_f(v(w_2+B))}{v(w_2-l)-v(w_2)} \geq 1$.

These two inequalities show that the background risk effect described by the mixed vector $(X, B)$ on the optimal prevention is ambiguous.
\end{remarca}

{\emph{Case 8}} (the comparison of $e^\ast$ and $e_8^\ast$)

The following proposition establishes a necessary and sufficient condition that adding a possibilistic background risk $A$ in period $1$ and
a probabilistic risk $Y$ in period $2$, the optimal prevention level raises.

\begin{propozitie}
The following assertions are equivalent:

(i)  $e^\ast \leq e^\ast_8$;

(ii) $\frac{E_f(u'(w_1-e^\ast+A))}{u'(w_1-e^\ast)}\leq \frac{M(v(w_2-l+Y))-M(v(w_2+Y))}{v(w_2-l)-v(w_2)}$.
\end{propozitie}

\begin{proof}
(4.12) and (3.2) are applied.
\end{proof}

\begin{remarca}
Assume that $u'''>0$, $v'''>0$, $E_f(A)=0$, $M(Y)=0$. Then, by Proposition 2.2 and \cite{courbage1}, the following inequalities hold:

$\frac{E_f(u'(w_1-e^\ast+A))}{u'(w_1-e^\ast)}\geq 1$, $\frac{M(v(w_2-l+Y))-M(v(w_2+Y))}{v(w_2-l)-v(w_2)} \geq 1$,

thus the effect of mixed background risk $(A, Y)$ on the level of optimal saving is ambiguous.
\end{remarca}

\section{Other cases of comparing optimal prevention levels}

In the previous section we studied the changes in optimal prevention when we pass from the benchmark model to the models where the background
risk appears in one of the periods or in both periods. The problem of changes in optimal saving can be posed in other cases too: one passes
from a situation in which the background risk is present in a single period (models 1, 2, 4, 5) to one of the situations when the background risk
is present in both periods (models 3, 6, 7, 8). For this purpose we have to compare the optimal prevention levels for the following pairs:
$(e_1^\ast, e_3^\ast)$, $(e_1^\ast, e_7^\ast)$, $(e_2^\ast, e_3^\ast)$, $(e_2^\ast, e_8^\ast)$, $(e_4^\ast, e_6^\ast)$, $(e_4^\ast, e_8^\ast)$,
$(e_5^\ast, e_6^\ast)$, $(e_5^\ast, e_7^\ast)$.

In this section we will deal with this problem making the comparison of optimal prevention levels of some of the mentioned pairs, the remaining cases being treated in a similar manner.

{\emph{Case 1}} (the comparison of $e_4^\ast$ and $e_6^\ast$)

One assumes that in period $1$ there exists a background risk described by the fuzzy number $A$, and to this is added in period $2$ a possibilistic background risk
described by a fuzzy number $B$. The following proposition offers a necessary and sufficient condition that  by this change, the optimal prevention level drops.

\begin{propozitie}
The following assertions are equivalent:

(i) $e^\ast_6 \leq e^\ast_4$;

(ii) $\frac{E_f(v(w_2-l+B))-E_f(v(w_2+B))}{v(w_2-l)-v(w_2)} \leq 1$.
\end{propozitie}

\begin{proof}
By (4.4), (4.9) and taking into account that $-E_f(u'(w_1-e^\ast_6+A))<0$, the following assertions are equivalent:

$\bullet$ $e^\ast_6 \leq e^\ast_4$;

$\bullet$ $V'_4(e^\ast_6) \geq V'_4(e^\ast_4)=0$;

$\bullet$ $-E_f(u'(w_1-e^\ast_6+A))+p'(e^\ast_6)[v(w_2-l)-v(w_2)]\geq 0$;

$\bullet$ $-E_f(u'(w_1-e^\ast_6+A))+E_f(u'(w_1-e^\ast_6+A))\frac{v(w_2-l)-v(w_2)}{E_f(v(w_2-l+B))-E_f(v(w_2+B))} \geq 0$;

$\bullet$ $-E_f(u'(w_1-e^\ast_6+A))[1-\frac{v(w_2-l)-v(w_2)}{E_f(v(w_2-l+B))-E_f(v(w_2+B))}] \geq 0$;

$\bullet$ $1-\frac{v(w_2-l)-v(w_2)}{E_f(v(w_2-l+B))-E_f(v(w_2+B))} \leq 0$;

$\bullet$ $\frac{E_f(v(w_2-l+B))-E_f(v(w_2+B))}{v(w_2-l)-v(w_2)} \leq 1$.

\end{proof}

\begin{corolar}
If $v'''>0$ and $E_f(B)=0$ then $e^\ast_6 \leq e^\ast_4$.
\end{corolar}

\begin{proof}
By Lemma 5.4, condition (ii) of Proposition 6.1 is fulfilled.
\end{proof}

The previous corollary shows that for a prudent agent (in period $2$), adding a possibilistic background risk $B$ to an already existing possibilistic background risk $A$ leads to the raise in the optimal prevention.

{\emph{Case 2}} (the comparison of $e_5^\ast$ and $e_6^\ast$)

One considers the model $5$ with a possibilistic background risk $B$ in period $2$. We intend to see how the optimal prevention changes if
we pass to model $6$, by adding a possibilistic background risk $A$ in period $1$.

\begin{propozitie}
The following assertions are equivalent:

(i) $e^\ast_6 \leq e^\ast_5$;

(ii) $u'(w_1-e^\ast_6) \leq E_f(u'(w_1-e^\ast_6+A))$.
\end{propozitie}

\begin{proof}
By (4.6) and (4.9) the following equivalences hold:

$e^\ast_6 \leq e^\ast_5$ iff $V'_5(e^\ast_6) \geq V'_5(e^\ast_5)=0$

\hspace{1.2cm} iff $-u'(w_1-e^\ast_6)+p'(e^\ast_6)[E_f(v(w_2-l+B))-E_f(v(w_2+B))]\geq 0$

\hspace{1.2cm} iff $-u'(w_1-e^\ast_6)+E_f(u'(w_1-e^\ast_6+A))\geq 0$

\hspace{1.2cm} iff $u'(w_1-e^\ast_6) \leq E_f(u'(w_1-e^\ast_6+A))$.

\end{proof}

\begin{corolar}
If $u'''>0$ and $E_f(A)=0$ then $e^\ast_6 \leq e^\ast_5$.
\end{corolar}

\begin{proof}
By applying Proposition 2.2, from $u'''>0$ and $E_f(A)=0$ it follows that condition (iii) of Proposition 6.3 is verified.
\end{proof}

{\emph{Case 3}} (the comparison of $e_1^\ast$ and $e_7^\ast$)

One considers model $1$ with a probabilistic background risk in period $1$ and one passes to model $7$ by adding a possibilistic background risk in period $2$.

\begin{propozitie}
The following assertions are equivalent:

(i) $e^\ast_7 \leq e^\ast_1$;

(ii) $\frac{E_f(v(w_2-l+B))-E_f(v(w_2+B))}{v(w_2-l)-v(w_2)} \leq 1$
\end{propozitie}

\begin{proof}
Recall from \cite{courbage1} (or from Section $3$, Model $1$) that

$V'_1(e)=-M(u'(w_1-e+X))+p'(e)[v(w_2-l)-v(w_2)]$.

Then, by (4.11) the following assertions are equivalent:

$\bullet$  $e^\ast_7 \leq e^\ast_1$

$\bullet$ $V'_1(e^\ast_7) \geq V'_1(e^\ast_1)=0$

$\bullet$ $-M(u'(w_1-e^\ast_7+X))+p'(e^\ast_7)[v(w_2-l)-v(w_2)]\geq 0$

$\bullet$ $-M(u'(w_1-e^\ast_7+X))[1-\frac{v(w_2-l)-v(w_2)}{E_f(v(w_2-l+B))-E_f(v(w_2+B))}] \geq 0$

$\bullet$ $\frac{E_f(v(w_2-l+B))-E_f(v(w_2+B))}{v(w_2-l)-v(w_2)} \leq 1$
\end{proof}

\begin{corolar}
If $v'''>0$ and $E_f(B)=0$ then $e^\ast_7 \leq e^\ast_1$.
\end{corolar}

\begin{proof}
By \cite{courbage1}, if $v'''>0$ and $E_f(B)=0$ then condition (ii) of Proposition 6.5 is fulfilled.
\end{proof}

\section{Concluding remarks}

The five models from the paper deal with the optimal prevention in the presence of some risk situations possibilistically described (by fuzzy numbers)
or mixed (by combinations of random variables and fuzzy numbers). They come in addition to probabilistic models of optimal saving developed by Courbage and Rey in \cite{courbage1}. The main results from the paper establish necessary and sufficient conditions for the optimal prevention increase or decrease when we add a possibilistic or mixed risk to the benchmark model.

Next we will state several research topics which could be investigated in line to this paper.

(I) Among the entities which define the benchmark model from Section 4, the probability $p(e)$ has an essential role. Also, all the eight background risk models from Section $5$ contain this probability as the defining element. Accordingly, even if the background risk is possibilistic or mixed, in the presence of $p(e)$, each model still preserves a probabilistic trace.

We could modify the nature of $p(e)$ in the following way:

Case $A$: "$p(e)$ is the possibility that a loss $l$ occurs at the prevention level $e$"

Case $B$: "$p(e)$ is the credibility \footnote{For the definition of credibility and its properties, see the monograph \cite{liu} or Chapter 8 from \cite{georgescu1}.} that a loss $l$ occurs at the prevention level $e$"

Study the optimal prevention in models in which $p(e)$ is a possibility (Case $A$) and a credibility (Case $B$), respectively.

(II) Both the definition of prevention models (total utility function, the first--order condition, etc.), and the optimal prevention level comparison theorems were formulated in terms of possibilistic expected utility, probabilistic expected utility or both of them. On the other hand, in the literature there exists a second notion of possibilistic expected utility (see for example \cite{georgescu1}, Section 4.2), as well as the concept of expected utility operator (see \cite{georgescu1}, Chapter 5), which comprises both of them.

Generalize the prevention models from the paper in a general context founded on expected utility operators.

(III) In order to reduce the action of various types of risk, an agent has within his reach several instruments: saving, self--protection and self--insurance (according to \cite{heinzel}). A {\emph{precautionary motive}} of an agent is defined in \cite{heinzel} as "the extent to which the agent
chooses an instrument for the response to a future risk".

In paper \cite{heinzel} it is analyzed the way the intensity of precautionary motives is influenced by the choice of various combinations of instruments
as well as the interaction of considered types of risk.

An open problem is to develop models of precautionary saving analysis corresponding to combinations of the same three instruments (saving, self--protection and self--insurance), but in which the risk parameters are probabilistic, possibilistic or mixed.

\end{document}